\documentclass[aps,prl,twocolumn, 10pt]{revtex4-2}
\usepackage[T1]{fontenc}
\usepackage{orcidlink}
\usepackage{amsmath,amssymb,amsthm}
\usepackage{bm}
\usepackage{mathrsfs}
\usepackage{xcolor}
\usepackage{hyperref}
\usepackage{nicefrac}
\usepackage{tikz}
\usetikzlibrary{arrows.meta,positioning,calc}
\newcommand{\Hp}{\mathscr H_{+}}
\newcommand{\Hm}{\mathscr H_{-}}
\newcommand{\Hd}{\mathscr H}
\newcommand{\Herm}{\mathfrak{herm}}

\newtheorem{theorem}{Theorem}
\newtheorem{lemma}[theorem]{Lemma}
\newtheorem{corollary}[theorem]{Corollary}
\theoremstyle{definition}
\newtheorem{definition}{Definition}
\begin{document}
\title{The Dirac Information Carrier for Relativistic Quantum Computation}
\author{Barry C. Sanders\,\orcidlink{0000-0002-8326-8912}}
\affiliation{Institute for Quantum Science and Technology, University of Calgary, Alberta T2N~1N4, Canada}
\begin{abstract}
Quantum computation has traditionally been formulated by postulating abstract
information carriers and subsequently identifying physical systems that realize
them. We adopt the opposite viewpoint and ask what computational structure
is supplied by a fundamental relativistic quantum system itself. 
Focusing on
the simplest nontrivial massive spin carrier, spin-$\nicefrac12$, we show that
its relativistic description supplies a native four-dimensional information
carrier through the Dirac equation.
The resulting Dirac information carrier
possesses an intrinsic positive- and negative-energy decomposition that induces
a physics-constrained computational structure comprising sector-preserving
and sector-coupling quantum logic.
Restricting the relativistic computational structure to the positive-energy
sector and taking the nonrelativistic regime recovers the familiar Pauli qubit
description.
We distinguish mathematical unitary transformations from physically admissible
gates and show how charge-conjugation structure, charge superselection, and
associated reference-frame resources can constrain quantum logic beyond the
sector-preserving structure.
Finally, under appropriate physical assumptions, we establish the conditions
under which this structure supports full single-carrier controllability.
Our work establishes the Dirac information carrier as the massive
spin-$\nicefrac12$ instance of a broader programme in which computational
structures are derived from the relativistic representation carried by the
underlying physical system.
\end{abstract}
\maketitle
Quantum computation is conventionally formulated by first postulating an
abstract computational primitive, such as a qubit or qudit, and then
identifying a physical system in which that primitive can be encoded and
controlled~\cite{NielsenChuang2010,WangHuSandersKais2020}. This separation
between computational abstraction and physical realization has enabled quantum
information processing to develop largely independently of the underlying
physical theory. A complementary first-principles question is whether the
physical theory itself can determine the native information carrier and the
computational structure acting upon it.

Relativistic quantum information typically asks how motion, reference frames,
and quantum fields modify the description and processing of quantum
information~\cite{PeresTerno2004}. Here we ask a different question: given a
massive relativistic quantum system characterized by its mass and spin, what
information-bearing structure follows from the corresponding relativistic
representation, and what quantum logic does that structure support? We address
this question for a free massive spin-$\nicefrac12$ system. In this case the
Dirac equation provides the covariant one-particle spinor description, and its
internal structure induces the computational structure studied below.

For a massive spin-$\nicefrac12$ particle, the relevant relativistic structure
is not selected by an algebraic analogy with the Pauli matrices but by the
physical representation carried by the system. The Dirac equation provides the
standard first-order covariant realization of this spinor structure and
necessarily introduces a four-component wavefunction together with positive-
and negative-energy solution sectors~\cite{Dirac1928}.
The Klein--Gordon
mass-shell relation~\cite{Gordon1926}, by contrast, does not by itself determine the internal
spin representation: fields of different spin can satisfy the corresponding
second-order mass-shell equation once the appropriate representation and
constraints are supplied~\cite{BargmannWigner1948}.
Our construction therefore begins with the massive
spin-$\nicefrac12$ Dirac carrier because that is the physical system under
consideration, not because the Dirac matrices happen to furnish a convenient
four-dimensional algebra. Unlike an abstract ququart or other four-level
qudit~\cite{WangHuSandersKais2020}, the resulting carrier possesses an
intrinsic positive- and negative-energy decomposition dictated by the
relativistic theory, which in turn determines the computational structure
developed here.

We identify the \emph{Dirac information carrier} as the native computational
primitive associated with the massive spin-$\nicefrac12$ relativistic system
considered here.
The one-particle Dirac Hilbert space can be resolved into momentum fibres, with each momentum~$\bm p$ associated with a four-dimensional internal spinor space~$\mathscr{H}_{\bm{p}}$.
An exactly sharp momentum is understood in the usual generalized-eigenstate sense~\cite{deLaMadrid2005}, whereas physical states are normalizable wavepackets over these fibres.
At each momentum, the Dirac Hamiltonian induces the intrinsic decomposition
\begin{equation}
\label{eq:decomposition}
\mathscr H_{\bm p}
=
\mathscr H_{+,\bm p}\oplus\mathscr H_{-,\bm p}
\end{equation}
into positive- and negative-energy sectors.
Henceforth we work within a fixed momentum fibre and suppress the momentum
label, writing~$\Hd=\Hp\oplus\Hm$.
The quantum logic developed below acts on this internal spinor space and is
therefore formulated fibrewise.

For physical wavepackets, this description
applies when the implemented control acts fibrewise, or as an effective
internal transformation over the momentum support of the state; interactions
that substantially couple different momentum fibres require treatment on the
full one-particle Hilbert space.
This structure is inherited from the
relativistic theory itself rather than chosen as an encoding for quantum
information processing. The resulting formulation distinguishes three
conceptual levels: the underlying relativistic physics described by the Dirac
equation, the information carrier determined by its four-component spinor
structure, and the computational structure induced by that carrier.
This Letter develops these three levels in turn, derives the computational
structure implied by the Dirac information carrier, establishes its
computational consequences, and finally examines the physical admissibility of
the resulting quantum logic.

The preceding discussion identified the Dirac information carrier and
previewed its intrinsic structure. We now formulate these notions precisely.
\begin{definition}
For the massive spin-\nicefrac12 system considered here,
a \emph{Dirac information carrier} is a quantum information-bearing free
massive spin-$\nicefrac12$ relativistic system whose one-particle states are
described by the Dirac equation.
Its computational realization within a
momentum fibre is the \emph{Dirac ququart}, namely the four-dimensional
internal spinor space endowed with the additional algebraic and physical
structure inherited from relativistic quantum mechanics.
\end{definition}
\noindent
For a free massive spin-$\nicefrac12$ particle, the Dirac equation is
\begin{equation}
\left(\text{i}\gamma^\mu\partial_\mu-m\right)\psi=0,\,
\{\gamma^\mu,\gamma^\nu\}
=2\eta^{\mu\nu}\mathbb I,\,
\hbar\equiv1\equiv\text{c},
\label{eq:Dirac}
\end{equation}
with Minkowski metric
$\eta=\operatorname{diag}(1,-1,-1,-1)$~\cite{Dirac1928,Thaller1992}.
Within the fixed momentum fibre, the Dirac Hamiltonian
\begin{equation}
H_{\rm D}(\bm p)=\bm\alpha\cdot\bm p+\beta m,\,
\alpha^j=\gamma^0\gamma^j,\,
\beta=\gamma^0
\end{equation}
satisfies
\begin{equation}
\label{eq:DiracHamiltonian}
H_{\rm D}(\bm p)^2=E_{\bm p}^2\mathbb I,\,
E_{\bm p}=\sqrt{\bm p^2+m^2}.
\end{equation}
Its spectral projectors
\begin{equation}
\label{eq:spectral-projectors}
\Pi_\pm(\bm p)
=\frac12\left(
\mathbb I\pm\frac{H_{\rm D}(\bm p)}{E_{\bm p}}
\right)
\end{equation}
define the two-dimensional positive- and negative-energy subspaces
$\mathscr{H}_\pm
=\operatorname{im}\Pi_\pm(\bm p)$, yielding the intrinsic
decomposition~\eqref{eq:decomposition}~\cite{Thaller1992}.
This decomposition is inherited from the relativistic theory itself and is
the physical origin of the computational structure developed below.

The positive- and negative-energy subspaces introduced here are spectral
subspaces of the first-quantized Dirac Hamiltonian.
For a fundamental Dirac
field, they should not be interpreted as physical particle states of positive
and negative energy: after second quantization, the negative-frequency branch
of the Dirac solution space is instead associated with positive-energy
antiparticle excitations~\cite{BjorkenDrell1965}.
Thus, the decomposition~\eqref{eq:decomposition}
defines the intrinsic $2+2$ structure of the Dirac solution space, whereas its
particle--antiparticle interpretation belongs to the corresponding quantized
field theory. The physical realization of transformations coupling these
sectors is considered below.

The relativistic transformation properties of the Dirac information carrier
are determined by the spin representation of $\operatorname{Spin}^{+}(1,3)$,
the double cover of the proper orthochronous
Lorentz group~\cite{LawsonMichelsohn1989,Thaller1992},
which determines how its
solutions transform between inertial reference frames. Once the Dirac spinor space is regarded as a quantum information carrier,
standard quantum kinematics assigns its traceless Hermitian endomorphisms as
the mathematical generators of single-carrier unitary control~\cite{NielsenChuang2010,DAlessandro2008}; which of these generators can be
realized physically is a separate question addressed below.
Writing $\sigma^{\mu\nu}:=\text{i}[\gamma^\mu,\gamma^\nu]/2$ and
$\gamma^5:=\text{i}\gamma^0\gamma^1\gamma^2\gamma^3$, and adopting the
standard Hermiticity convention
$(\gamma^0)^\dagger=\gamma^0$ and
$(\gamma^j)^\dagger=-\gamma^j$, the complexified spacetime Clifford algebra
provides the following basis of the fifteen-dimensional real vector space of
traceless Hermitian endomorphisms of~$\Hd$~\cite{Lounesto2001}:
\begin{equation}
\mathcal G
=\left\{
\gamma^0,\;
\text{i}\gamma^j,\;
\text{i}\sigma^{0j},\;
\sigma^{jk},\;
\gamma^5,\;
\text{i}\gamma^5\gamma^0,\;
\gamma^5\gamma^j\right\}
\end{equation}
for $j,k\in\{1,2,3\}$
with $j>k$.

Writing~$\mathcal G=\{G_a\}_{a=1}^{15}$, every traceless Hermitian
operator on the Dirac information carrier has the expansion
\begin{equation}
\label{eq:Hexpansion}
H=\sum_{a=1}^{15}\theta_aG_a,\qquad
\theta_a\in\mathbb R .
\end{equation}
As~$-\text{i}H\in\mathfrak{su}(4)$, exponentiation yields the unitary
computational transformation
\begin{equation}
\label{eq:unitaryquantumlogic}
U=\exp(-\text{i}H)\in\text{SU}(4).
\end{equation}

Together with~$\mathbb I$, the fifteen matrices in~$\mathcal G$ form a basis
of the full complex endomorphism algebra of the Dirac ququart, consistent with
\begin{equation}
\label{eq:complexifiedspacetimeClifford}
\mathbb C\otimes\mathrm{Cl}_{1,3}
\cong
M_4(\mathbb C)
\cong
\operatorname{End}(\Hd).
\end{equation}
Thus the complexified spacetime Clifford algebra intrinsic to the Dirac theory
furnishes a natural complete basis for the matrix algebra acting on the Dirac
information carrier, whereas the traceless Hermitian subspace generates
$\mathrm{SU}(4)$ after multiplication by~$-\text{i}$ and exponentiation~\cite{Lounesto2001}.
Relativistic covariance
and quantum logic therefore arise from distinct structures acting on the same
underlying vector space.
The former is governed by the spin representation of
the Lorentz group, whereas the latter is generated by traceless Hermitian endomorphisms
whose multiplication by~$-\text{i}$ and exponentiation yield the computational symmetry group~$\text{SU}(4)$.

The passage from relativistic covariance to quantum logic consequently does
not identify the Lorentz group with~$\text{SU}(4)$; instead, it proceeds
through the traceless Hermitian endomorphisms acting on the Dirac information carrier.
Accordingly, we define
\begin{definition}
Let
\begin{equation}
\label{eq:Herm0}
\Herm_0(\Hd)
:=
\left\{
H\in\operatorname{End}(\Hd):
H^\dagger=H,\,
\operatorname{tr}H=0
\right\}
\end{equation}
denote the real vector space of traceless Hermitian endomorphisms acting on the
Dirac information carrier. Multiplication by~$-\text{i}$ identifies
$\Herm_0(\Hd)$ with~$\mathfrak{su}(4)$, and exponentiation
$H\longmapsto\exp\left(-\text{i}H\right)$
produces unitary computational transformations in~$\text{SU}(4)$.
\end{definition}
\noindent
Denoting the Dirac information carrier by~$\mathsf D$, the construction is
summarized schematically by
\begin{equation}
\operatorname{Spin}^{+}(1,3)
\curvearrowright
\mathsf D
\longmapsto
\Herm_0(\Hd)
\xrightarrow{\exp(-\text{i}\,\cdot)}
\text{SU}(4),
\label{eq:covariance-to-logic}
\end{equation}
where $\curvearrowright$ denotes the action of the spin representation on
the Dirac information carrier, $\longmapsto$ denotes the passage from the
carrier to the traceless Hermitian endomorphisms generating computational
transformations, and
$\xrightarrow{\exp(-\text{i}\,\cdot)}$ denotes exponentiation to the
computational symmetry group~$\text{SU}(4)$.

The unitary group~$\text{SU}(4)$ therefore provides the mathematical
symmetry of quantum logic acting on the Dirac ququart. This symmetry,
however, describes the space of unitary transformations on the computational
Hilbert space~$\Hd$ rather than the relativistic covariance of the underlying
physical system. We therefore distinguish mathematical transformations from
their physical realization.

\begin{definition}
An \emph{admissible quantum logic gate} is a unitary transformation
on the Dirac information carrier that is physically implementable
under the stated physical assumptions for the realization under consideration.
\end{definition}
\noindent
Accordingly, not every element of~$\text{SU}(4)$ need correspond to an
admissible quantum logic gate.
Relative to the intrinsic decomposition~(\ref{eq:decomposition}),
admissible gates are classified as
\emph{sector-preserving}, acting independently on~$\Hp$ and~$\Hm$, or
\emph{sector-coupling}, coherently connecting the two sectors. This
classification is determined mathematically by the block structure of the
corresponding unitary operators, whereas the physical admissibility of these
transformations depends on additional physical principles considered below.

The physical admissibility of a sector-coupling transformation depends on the
realization of the Dirac information carrier.
At the level of the first-quantized Dirac equation, charge conjugation relates
the positive- and negative-energy solution structures. Its action on
c-number Dirac spinors involves complex conjugation and is therefore
antilinear, so it is not itself an element of the complex-linear
single-carrier group~$\text{SU}(4)$ considered here.
After second quantization,
charge conjugation is instead represented by a unitary operator on Fock space
that interchanges particle and antiparticle degrees of freedom~\cite{BjorkenDrell1965}.
This Fock-space symmetry should not be identified with arbitrary
single-carrier~$\text{SU}(4)$ control between the first-quantized energy sectors.

For a fundamental charged Dirac field, the particle--antiparticle
interpretation described above brings charge conservation and charge
superselection into the question of physical admissibility: coherent
transformations involving sectors of different charge can require an
appropriate reference-frame or frameness resource.
This operational viewpoint
is consistent with resource-theoretic treatments of charge and discrete
symmetries~\cite{GourSandersTurner2009,SkotiniotisTolouiDurhamSanders2013,SkotiniotisTolouiDurhamSanders2014}.
The required resources and their physical realization are platform dependent;
an effective quantum simulator of the Dirac structure, for example, need not
inherit the superselection restrictions of a fundamental charged field. This
Letter therefore distinguishes the mathematical classification of
sector-preserving and sector-coupling quantum logic from the
realization-dependent physical admissibility of those transformations.

Full controllability of a single Dirac information carrier is conventionally
formulated in terms of finite gate sets. Accordingly, we regard an admissible
gate set as a finite collection of admissible quantum logic gates acting on
the Dirac information carrier. Such a gate set is said to provide
\emph{full single-carrier controllability} if the subgroup that it generates
is dense in the target unitary group with respect to the operator-norm
topology~\cite{NielsenChuang2010,BarencoEtAl1995,BoykinEtAl2000}.
We say that \emph{sector-preserving controllability} is available when the
admissible sector-preserving gates generate a dense subgroup of the full
sector-preserving control group. Assuming sector-preserving controllability,
arbitrary single-qubit control is already available independently within the
positive- and negative-energy sectors. The essential question is therefore
under what additional conditions an admissible finite gate set achieves full
relativistic controllability of the Dirac information carrier. The following
theorem answers this question.

Relative to the decomposition~(\ref{eq:decomposition}), 
the sector-preserving subgroup is~\cite{Hall2015,Knapp2002}
\begin{equation}
\label{eq:sector-preserving subgroup}
G_{\text{sp}}
:=
S\!\left(\text{U}(2)\times\text{U}(2)\right)
\subset\text{SU}(4).
\end{equation}
The subgroup~$G_{\text{sp}}$ is connected, with Lie
algebra~$\mathfrak g_{\text{sp}}$ comprising the block-diagonal elements
of~$\mathfrak{su}(4)$, 
whereas the complementary subspace
\begin{equation}
\label{eq:complementarysubspace}
\mathfrak p
=
\left\{
P(Z):=
\begin{pmatrix}
0&Z\\
-Z^\dagger&0
\end{pmatrix}
:
Z\in M_2(\mathbb C)
\right\}
\end{equation}
comprises the off-diagonal elements, so that~\cite{Knapp2002}
\begin{equation}
\label{eq:su4gsp+p}
\mathfrak{su}(4)=\mathfrak g_{\text{sp}}\oplus\mathfrak p.
\end{equation}
\begin{lemma}
The subgroup~$G_{\text{sp}}$ is a maximal connected proper subgroup
of~$\text{SU}(4)$.
\end{lemma}
\begin{proof}
For
\begin{equation}
\label{eq:gdiag}
g=\operatorname{diag}(U_+,U_-)\in G_{\text{sp}},
\end{equation}
direct block-matrix multiplication gives the adjoint action, i.e., conjugation, on
$\mathfrak p$ is
\begin{equation}
\label{eq:adjaction}
P(Z)\longmapsto P(U_+ZU_-^\dagger).
\end{equation}
The central one-parameter subgroup
\begin{equation}
C(t)=
\operatorname{diag}
\!\left(
\text{e}^{\text{i}t}\mathbb I_2,
\text{e}^{-\text{i}t}\mathbb I_2
\right)
\end{equation}
acts as~$Z\mapsto \text{e}^{2\text{i}t}Z$. Hence every real
$G_{\text{sp}}$-invariant subspace of~$\mathfrak p$ is also complex.
As a complex representation of
$\text{SU}(2)\times\text{SU}(2)$,
\begin{equation}
M_2(\mathbb C)
\cong
\mathbb C^2\otimes(\mathbb C^2)^*
\end{equation}
is irreducible, so~$\mathfrak p$ is irreducible as a real
$G_{\text{sp}}$-module.

Now let
\begin{equation}
\mathfrak g_{\text{sp}}\subsetneq\mathfrak k\subseteq\mathfrak{su}(4).
\end{equation}
Exploiting the fact that every element has a unique decomposition~(\ref{eq:su4gsp+p}),
choose
\begin{equation}
\label{eq:XDP}
X=D+P\in\mathfrak k\setminus\mathfrak g_{\text{sp}},\,
D\in\mathfrak g_{\text{sp}}
\end{equation}
and nonzero~$P\in\mathfrak p$, and set
\begin{equation}
J=
\text{i}\operatorname{diag}(\mathbb I_2,-\mathbb I_2)
\in\mathfrak g_{\text{sp}}.
\end{equation}
Then~$[J,X]=[J,P]$ is a nonzero element of
$\mathfrak k\cap\mathfrak p$, because~$\operatorname{ad}_J$ is invertible
on~$\mathfrak p$. Irreducibility therefore gives
$\mathfrak p\subseteq\mathfrak k$, whence
$\mathfrak k=\mathfrak{su}(4)$. Thus no connected proper subgroup lies
strictly between~$G_{\text{sp}}$ and~$\text{SU}(4)$
~\cite{Dynkin1952,Knapp2002,Hall2015}.
\end{proof}

Let~$\mathcal S_{\text{sp}}$ be a finite admissible sector-preserving gate
set satisfying
\begin{equation}
\overline{\left\langle\mathcal S_{\text{sp}}\right\rangle}
=
G_{\text{sp}}.
\end{equation}
This assumption includes complete control within each sector and control of
the relative-sector phase. We use controllability in the standard quantum-control sense of generating a
dense subgroup of the target unitary group~\cite{DAlessandro2008}.
A sector-coupling gate need not enlarge this
closure: for example, a gate could exchange~$\Hp$ and~$\Hm$ while preserving
their decomposition as an unordered pair. The relevant criterion is therefore
membership in the normalizer
\begin{equation}
N_{\text{SU}(4)}\!\left(G_{\text{sp}}\right)
:=
\left\{
V\in\text{SU}(4):
VG_{\text{sp}}V^\dagger=G_{\text{sp}}
\right\}.
\end{equation}
The normalizer is a proper closed subgroup whose identity component is
$G_{\text{sp}}$; its other component exchanges~$\Hp$ and~$\Hm$.
Indeed, an element normalizing $G_{\mathrm{sp}}$ must preserve the two
two-dimensional sectors or exchange them as an unordered pair.

\begin{theorem}[Full Relativistic Controllability Theorem]
Let~$\mathcal S_{\text{sp}}$ be a finite admissible sector-preserving gate
set whose generated closure is~$G_{\text{sp}}$, and let~$U_{\text{c}}$ be
an additional admissible quantum logic gate.
Then
$\mathcal S_{\text{sp}}\cup\{U_{\text{c}}\}$ provides full single-carrier
controllability of the Dirac information carrier if and only if
\begin{equation}
U_{\text{c}}
\notin
N_{\text{SU}(4)}\!\left(G_{\text{sp}}\right).
\end{equation}
Equivalently,
\begin{equation}
\overline{
\left\langle
\mathcal S_{\text{sp}},U_{\text{c}}
\right\rangle}
=
\text{SU}(4)
\quad\Longleftrightarrow\quad
U_{\text{c}}
\notin
N_{\text{SU}(4)}\!\left(G_{\text{sp}}\right).
\end{equation}
\end{theorem}
\begin{proof}
Let
\begin{equation}
K:=
\overline{
\left\langle
\mathcal S_{\text{sp}},U_{\text{c}}
\right\rangle}.
\end{equation}
Let $K^\circ$ denote the connected component of $K$ containing the identity.
Since
$G_{\mathrm{sp}}
=\overline{\langle\mathcal S_{\mathrm{sp}}\rangle}\subseteq K$
is connected and contains the identity, we have
$G_{\mathrm{sp}}\subseteq K^\circ$.
The maximality of
$G_{\mathrm{sp}}$ among connected proper subgroups of~$\text{SU}(4)$
therefore implies
$K^\circ=G_{\mathrm{sp}}$ or~$K^\circ=\text{SU}(4)$.

If
$U_{\text{c}}\notin N_{\text{SU}(4)}(G_{\mathrm{sp}})$, then
$K^\circ\neq G_{\mathrm{sp}}$: otherwise, as the identity component
$K^\circ$ is normal in~$K$ and~$U_{\text{c}}\in K$, we would have
$U_{\text{c}}G_{\mathrm{sp}}U_{\text{c}}^\dagger=G_{\mathrm{sp}}$,
contradicting the assumption. Hence~$K^\circ=\text{SU}(4)$ and therefore
$K=\text{SU}(4)$.
Conversely, if
$U_{\text{c}}\in N_{\text{SU}(4)}(G_{\text{sp}})$, then every element
generated by~$\mathcal S_{\text{sp}}\cup\{U_{\text{c}}\}$ lies in this
normalizer.
As
$N_{\text{SU}(4)}(G_{\mathrm{sp}})$ is a proper closed subgroup
of~$\text{SU}(4)$~\cite{Knapp2002},
the generated closure cannot equal~$\text{SU}(4)$.
\end{proof}
\begin{corollary}
Full relativistic controllability requires an admissible quantum
logic gate outside the normalizer of the sector-preserving subgroup. Thus,
preserving the positive- and negative-energy sectors, or merely exchanging
them, is insufficient to access the full computational symmetry of the Dirac
information carrier.
\end{corollary}
\noindent
For multiple Dirac information carriers, scalable universal quantum
computation would additionally require an appropriate admissible entangling
interaction between carriers, as in the standard theory of universal quantum
computation~\cite{BarencoEtAl1995}.

The Full Relativistic Controllability Theorem identifies the precise
mathematical condition under which the computational structure of the
Dirac information carrier is fully accessible.
The intrinsic decomposition
into positive- and negative-energy sectors determines a sector-preserving
computational structure,
but does not by itself restrict quantum logic to
that subgroup.
Whether transformations beyond it are physically realizable as quantum logic gates depends on admissibility:
sector-coupling gates require operational resources
capable of establishing the corresponding coherence. Full relativistic
controllability therefore depends not merely on the dimensionality of the
Dirac information carrier, but on which transformations permitted by its
computational symmetry can be realized as quantum logic gates.

The construction summarized by Eq.~\eqref{eq:covariance-to-logic} clarifies the
logical relationship between relativistic quantum mechanics and quantum
computation.
The passage from relativistic quantum mechanics to quantum logic
proceeds through the traceless Hermitian endomorphisms acting on the Dirac
information carrier rather than through the relativistic covariance group itself.
Relativistic covariance determines how the Dirac information carrier
transforms between inertial reference frames, whereas quantum logic is
constructed from the traceless Hermitian endomorphisms acting upon that
carrier.
The computational symmetry therefore does not replace or extend the
relativistic symmetry; instead, it is constructed from a distinct mathematical
structure acting on the same underlying information carrier.
This separation
of relativistic covariance, computational structure, and admissible quantum
logic provides the conceptual foundation for relativistic quantum computation
developed here.

Our framework clarifies the status of the conventional Pauli
qubit within relativistic quantum computation. Restriction to the
positive-energy sector reduces the Dirac ququart to a two-dimensional
information carrier; in the nonrelativistic regime, this sector reduces to
the familiar Pauli spinor and hence the conventional spin qubit
description~\cite{FoldyWouthuysen1950}, explaining its effectiveness for the
vast majority of existing quantum information
processing~\cite{NielsenChuang2010}.
The relativistic
framework developed here does not replace this description; rather, it embeds
it within a more general computational structure whose intrinsic
positive- and negative-energy decomposition is inherited directly from the
Dirac equation.
For the massive spin-$\nicefrac12$ system considered here, relativistic quantum
computation is therefore obtained not by modifying the conventional qubit, but
by deriving the computational structure from the underlying relativistic
theory and exploiting the full structure of the resulting Dirac information
carrier when physically admissible.

The framework developed here naturally suggests several directions for future
research.
First, the operational characterization of admissible quantum logic
gates motivates a resource-theoretic formulation of relativistic quantum
computation in which access to sector coherence is characterized by the
reference-frame resources required for its realization.
Second, the distinction between computational symmetry and physical
admissibility provides a foundation for investigating relativistic fault tolerance and error
correction~\cite{LidarBrun2013}.
Third, deriving computational structures directly from fundamental physical
theories might lead to new classes of relativistic quantum algorithms. 
More generally,
the appropriate relativistic information carrier should be
determined by the mass, spin, and corresponding Poincaré representation of the
underlying physical system~\cite{Wigner1939}.
The Dirac carrier studied here is the massive
spin-$\nicefrac12$ instance of this broader programme; higher-spin systems
should lead to distinct information carriers and computational structures
once their relativistic representation structure and physical constraints are
taken into account~\cite{BargmannWigner1948}. This perspective suggests that
computational structures should be derived from the intrinsic
information-bearing structures supplied by fundamental physical theories
rather than imposed independently of them.

\begin{acknowledgments}
\paragraph{Acknowledgments:}
This work was supported by the Natural Sciences and Engineering Research
Council of Canada (NSERC) Discovery Grant \emph{High-dimensional quantum
computing}.
I thank Aninda Sinha for insightful comments on an earlier version of this
manuscript that prompted a substantial clarification of its mathematical
structure, Carlo Maria Scandolo for valuable comments that led to further
clarification of its mathematical and operational formulation, and Hubert de
Guise for perceptive comments that sharpened the physical motivation, scope,
and representation-theoretic framing of the work.
Assistance with manuscript development and editing was provided by ChatGPT
(OpenAI GPT-5.6 Sol).
Responsibility for all scientific content rests
entirely with the author, who independently verified every mathematical and
scientific statement. This work was conducted on the traditional Treaty~7
territory in Southern Alberta.
\end{acknowledgments}
\bibliography{references}
\end{document}